\documentclass[12pt]{llncs}
\usepackage{fullpage}
\usepackage{amsmath}
\usepackage{caption}
\usepackage[utf8]{inputenc}
\usepackage[T1]{fontenc}
\usepackage{amsfonts}
\usepackage{amssymb}
\usepackage{linearb}
\usepackage{float}
\usepackage{url}
\usepackage[all]{xy}
\usepackage{textgreek}
\usepackage{dsfont}
\usepackage{stmaryrd}
\usepackage{pxfonts}
\usepackage{bbm}
\usepackage{multirow}
\usepackage{tikz}
\usepackage[fancy]{tikz-inet}
\usetikzlibrary{automata}
\usetikzlibrary{arrows}
\usetikzlibrary{shapes}
\usetikzlibrary{fit}
\usepackage{ifthen}
\usepackage{stmaryrd}

\tikzstyle{every picture}=[
  >=stealth', 
shorten >=1pt, 
node distance=1.44cm,
auto,
bend angle=45,
initial text=,
  every state/.style={inner sep=0.75mm, minimum size=1mm},
font=\small,
]
\def\cqfd{\hfill $\qed$}
\newcommand{\cf}[1]{#1_{\mathrm{sf}}}

\def\sc{\mathrm{sc}}
\def\ie{\emph{i.e.} }

\newcommand{\ugh}{\mathbin{\textlinb{\BPwheel}}}
\title{The state complexity of a class of operations involving  roots and boolean operations}

\author{Pascal Caron \thanks{Pascal.Caron@univ-rouen.fr}, Edwin Hamel-de-le-court
    \thanks{Edwin.Hamel-de-le-court@etu.univ-rouen.fr}, and
    Jean-Gabriel Luque\thanks{Jean-Gabriel.Luque@univ-rouen.fr}
}
\institute{LITIS, Université de Rouen,\\ Avenue de l'Université,\\ 76801 Saint-\'Etienne du Rouvray Cedex,\\ France}

\begin{document}
\maketitle

\begin{abstract}
Modifiers are a sets of functions acting on tuple of automata and allowing one  to construct regular operations. We define and study the class of  friendly modifiers that  describes a class of regular operations  involving  compositions of boolean operations and roots. We also  give an explicit tight bound for the state complexity of these operations.
\end{abstract}

\section{Introduction}
The state complexity of a regular language is the size of its minimal (complete deterministic) automaton and the state complexity of a regular operation is the maximal one of those languages obtained by applying this operation onto languages of fixed state complexities.
The research on this subject dates back to the $70s$, when in a seminal paper  \cite{Mas70}, Maslov gave values (without proofs) of the state complexities of several operations, in particular of the square root. Since the $90s$, this area of research has been very active and the state complexity of numerous operations has been computed, see \emph{e.g} \cite{Dom02,GMRY17,JJS05,Jir05,JO08,Yu01a} for a survey.

A method that applies to a wide class of operations, called \emph{1-uniform} operations, was described independently in \cite{CHLP20} and \cite{Dav18}. This approach consists  in describing a $1$-uniform operation as  functions, called \emph{modifiers}, acting on automata satisfying some nice conditions. It allows us to describe  states of the resulting automaton as combinatorial objects and to compute an upper bound for its state complexity by enumerating them.
 Therefore, we obtain a tight bound by exhibiting a witness chosen in a pool of automata, called  {\it monsters},  whose  sets of transition functions contain all possible functions of their states.
 
 Krawetz \emph{et al.} \cite{KLS05} computed the state complexity of $\mathbf{Root}(L)$, the set of  words $w$ such that $w^{i}$ is in L for some positive integer $i$, which is the same as the state complexity of square root \cite{CHLP20}. This naturally raises the  question of whether this result can be generalized for any compositions of  boolean operations and roots.

We show a correspondence between  these operations, a class of modifiers, called \emph{friendly}, and subsets of  eventually periodic boolean sequences.
   From this correspondence,  we deduce a tight  bound for the state complexities of these operations.\\
The paper is organized as follows. Section \ref{sect-prel} contains definitions and notations about automata. In Section \ref{sec-mod}, we recall  definitions and  basic facts on modifiers and $1$-uniform operations. In Section \ref{sect-friend}, we define and study friendly modifiers and the operations they describe, in particular by introducing the notions of standard modifier and characteristic functions. Our main result is a one to one correspondence between standard friendly modifiers and operations, also called friendly, obtained by combining roots and boolean operations. Finally, In Section \ref{sec-sc}, we give  tight bounds for state complexities of friendly operations.
\section{Preliminaries\label{sect-prel}}
\subsection{Operations over sets and sequences}
The \emph{cardinality} of a finite set $E$ is denoted by $\#E$,  the \emph{set of subsets} of  $E$ is denoted by $2^E$ and the \emph{set of mappings} of $E$ into itself is denoted by $E^E$. The identity mapping from $E$ into itself is denoted $\mathrm{Id}_E$. For any non-negative integer $k$, we denote by $E^k$ the set of all $k$-tuples of elements of $E$. The \emph{symmetric difference} of two sets $E_1$ and $E_2$ is denoted by $\oplus$ and defined by $E_1\oplus E_2=(E_1\cup E_2)\backslash (E_1\cap E_2)$. For any positive integer $n$, let us denote $\llbracket n\rrbracket$ for $\{0,\ldots, n-1\}$. 

 For any two $k$-tuples of functions $\underline\phi=(\phi_1,\ldots,\phi_k)\in G_{1}^{F_1}\times\cdots\times G_{k}^{F_{k}}$ and $\underline\psi=(\psi_1,\ldots,\psi_k)\in F_{1}^{E_1}\times\cdots\times F_{k}^{E_{k}}$, we denote by $\underline\phi\circ\underline\psi=(\phi_1\circ\psi_1,\ldots,\phi_k\circ\psi_k)\in  G_{1}^{E_1}\times\cdots\times G_{k}^{E_{k}}$ the \emph{point by point composition} of $\underline\psi$ by $\underline\phi$. This operation should not be confused with the \emph{composition}  defined as follow: for any two functions $f: E^j\rightarrow E$ and $g:F^k\rightarrow E$ and for any $1\leq p\leq j$, the $p$th \emph{composition} of $g$ by $f$ is the 
 function  $f\circ_{p}g: E^{j+k-1}\rightarrow E$ defined by
\[f\circ_{p}g(e_{1},\dots,e_{j+k-1})=f(e_{1},\dots,e_{p-1},g(e_{p},\dots,e_{p+k-1}),e_{p+j},\dots,e_{j+k-1}),\]
for any $e_{1},\dots, e_{j+k-1}\in E$.

For any set $E$, a \emph{sequence} with values in $E$ is a function $u$ from $\mathbb N$ into $E$. For every $p\in\mathbb N$, we denote $u(p)$ by $u_p$, and the sequence $u$ is often denoted by $(u_p)_{p\in\mathbb N}$. A sequence $u$ is \emph{eventually periodic} if there exists $a,b\in\mathbb N$ such that, for any integer $p\geq a$, $g_{p+b}=g_p$.

\subsection{Languages and Automata}
Let $\Sigma$ be a finite set of symbols, called \emph{alphabet}, the elements of which are called \emph{letters}. A \emph{word} $w$ over $\Sigma$ is a finite sequence of letters of $\Sigma$. Let us denote by $\varepsilon$ the empty word. The catenation of two words $u=a_1\cdots a_n$ and $v=b_1\cdots b_m$ denoted by $u\cdot v$ or $uv$ is the word $a_1\cdots a_nb_1\cdots b_n$. We define $w^n$ inductively as $w\cdot w^{n-1}$ with $w^0=\varepsilon$.

The set of all finite words over $\Sigma$ is denoted by $\Sigma ^*$. A \emph{language} is a subset of $\Sigma^*$. We define the \emph{complementary} of a language $L\subseteq\Sigma^{*}$ by $\ ^{c}L=\Sigma^{*}\setminus L$. For any $n\in\mathbb N$, we also define the $n$-th root of a language $L$ by $\sqrt[n]{L}=\{w\in\Sigma^*\mid w^n\in L\}$. Notice that $\sqrt[0]{L}=\Sigma^{*}$ if $\varepsilon\in L$ and $\emptyset$ otherwise, and $\sqrt[1]L=L$. By convention, we denote $\sqrt[2] L$ by $\sqrt L$.

A  \emph{complete and deterministic finite automaton} (DFA) is a $5$-tuple $A=(\Sigma,Q,i,F,\delta)$ where $\Sigma$ is the input alphabet, $Q$ is a finite set of states, $i\in Q$ is the  initial states, $F\subset Q$ is the set of final states and $\delta$ is the transition function  from  $Q\times \Sigma$ to $Q$ and defined for every $q\in Q$ and every $a\in \Sigma$. 
 The cardinality of $A$ is the cardinality of its set of states, \ie $\#A=\#Q$. 

Let $A=(\Sigma,Q,i,F,\delta)$ be a DFA. A word $w\in \Sigma ^*$ is \emph{recognized} by the DFA $A$ if $\delta(i,w)\in F$. The \emph{language recognized} by a DFA $A$ is the set $\mathrm L(A)$ of words recognized by $A$. Two DFAs are said to be \emph{equivalent} if they recognize the same language.

For any word $w$, we denote by $\delta^w$ the function $q\rightarrow\delta(q,w)$. Two states $q_1,q_2$ of $D$ are \emph{equivalent} if for any word $w$ of $\Sigma^*$, $\delta(q_1, w)\in F$ if and only if $\delta(q_2, w)\in F$. This equivalence relation is called the \emph{Nerode equivalence} and is denoted by $q_1\sim_{Ner} q_2$. If two states are not equivalent, then they are called \emph{distinguishable}.

A state $q$ is \emph{accessible} in a DFA  if there exists a word $w\in \Sigma ^*$ such that $q=\delta(i,w)$. A DFA is  \emph{minimal} if there does not exist any equivalent  DFA  with less states and it is well known that for any DFA, there exists a unique, up to a relabeling of the states, minimal equivalent one \cite{HU79}. Such a minimal DFA  is  obtained from $D$ by computing $\widehat A_{/\sim}=(\Sigma,Q/\sim,[i],F/\sim,\delta_{\sim})$ where $\widehat A$ is the accessible part of $A$, and where, for any $q\in Q$, $[q]$ is the $\sim$-class of the state $q$ and satisfies the property  $\delta_{\sim}([q],a)=[\delta(q,a)]$, for any $a\in \Sigma$.

\subsection{Languages Operations and State Complexity}
We consider that a $k$-ary \emph{operation over languages} (for short an operation) is a map sending every $k$-tuple of languages over the same alphabet to a language over the same alphabet as its preimage.  A $k$-ary operation is \emph{regular} if it sends every  $k$-tuple of regular languages to a regular language.
The state complexity of a regular language $L$ denoted by $\sc(L)$ is the number of states of its minimal DFA. This notion extends to regular operations: the state complexity of a unary regular operation $\otimes$ is the function $\sc_{\otimes}$ such that, for all $n\in\mathbb N\setminus 0$, $\sc_{\otimes}(n)$ is the maximum of all the state complexities of $\otimes(L)$ when $L$ is of state complexity $n$, \ie $\sc_{\otimes}(n)=\max\{\mathrm{sc}(\otimes(L)) | \sc(L) = n\}$.
More generally, the state complexity of a $k$-ary operation $\otimes$ is the $k$-ary function $\sc_\otimes$ such that, for all $(n_1,\ldots,n_k)\in (\mathbb N\setminus 0)^k$, \ie  $\sc_\otimes(n_1,\ldots,n_k)=\max\{\sc(\otimes(L_1,\ldots,L_k)\mid\text{ for all }i\in\{1,\ldots,k\}, \sc(L_i)=n_i\}.$
A \emph{ witness }for $\otimes$ is a a way to assign to each $(n_1,\ldots,n_k)$, assumed sufficiently big, a $k$-tuple of languages $(L_1,\ldots,L_k)$, over the same alphabets, with $\mathrm{sc}(L_i)=n_i$, for all $i\in\{1,\ldots,k\}$, such that $\sc_\otimes(n_1,\ldots,n_k)=\sc(\otimes(L_1,\ldots,L_k))$.
\subsection{1-uniform morphisms}
Let $\Sigma$ and $\Gamma$ be two alphabets. A morphism is a function $\phi$ from $\Sigma^*$ to $\Gamma^*$ such that, for all $w,v\in\Sigma^*$, $\phi(wv)=\phi(w)\phi(v)$. Notice that $\phi$ is completely defined by its value on letters.

Let $L$ be a regular language recognized by the DFA $A=(\Gamma,Q,i,F,\delta)$ and let $\phi$ be a morphism from $\Sigma^*$ to $\Gamma^*$. Then, $\phi^{-1}(L)$ is the regular language recognized by the DFA $B=(\Sigma,Q,i,F,\delta')$ where, for all $a\in\Sigma$ and $q\in Q$, $\delta'(q,a)=\delta(q,\phi(a))$.

A morphism $\phi$ is \emph{$1$-uniform} if the image by $\phi$ of any letter is a letter.

\section{Modifiers and 1-uniform operations}\label{sec-mod}
\subsection{Definition and first properties}
In \cite{CHL19} and \cite{Dav18}, the  authors of the respective papers investigated, independently, the same class of regular operations which are especially handy to manipulate for computing their state complexities. We recall its definition below.
  \begin{definition}\label{def-uni}
  A $k$-ary regular operation $\otimes$ is $1$-uniform if, for any $k$-tuple of regular languages $(L_1,\ldots,L_k)$, for any $1$-uniform morphism $\phi$, $\otimes(\phi^{-1}(L_1),\ldots,\phi^{-1}(L_k))=\phi^{-1}(\otimes(L_1,\ldots,L_k))$.
\end{definition}
 We check easily that the $1$-uniformity is stable by composition.
  \begin{claim}
    Let $\phi$ and $\psi$ be two $1$-uniform operations, respectively  $j$-ary and  $k$-ary. Then, for any integer $1\leq p\leq j$, the  $(j+k-1)$-ary  operator
        \[\phi\circ_p\psi(L_1,\ldots,L_{j+k-1})=\phi(L_1,\ldots,L_{p-1},\psi(L_{p},\ldots,L_{p+k-1}),L_{p+k},\ldots,L_{j+k-1})\]
    is $1$-uniform.
  \end{claim}
  Many well-known unary regular operations are $1$-uniform. See \cite{Dav18} for a non-exhaustive list of examples like the complement, the Kleene star, the reverse, the cyclic shift, the mirror, all boolean operations and catenation among others.

Each $1$-uniform operation corresponds to a construction over DFAs, which is handy when we need to compute the state complexity of its elements. Such a construction on DFAs has some constraints that are described in the following definitions.
\begin{definition}
  The \emph{state configuration} of a DFA $A=(\Sigma,Q,i,F,\delta)$ is the triplet $(Q,i,F)$.
\end{definition}
\begin{definition}\label{def-mod}
  A \emph{$k$-modifier}  is a $k$-ary operation acting on a $k$-tuple of DFAs $(A_1,\ldots,A_k)$, on the same alphabet $\Sigma$, and  producing a DFA  $\mathfrak m (A_1,...,A_k)$  such that
  \begin{itemize}
  \item its alphabet is $\Sigma$,
  \item its  state configuration  depends only on the state configurations of the DFAs $A_1,\ldots,A_k$,
  \item for any letter $a\in\Sigma$, the transition function of $a$ in $\mathfrak m (A_1,\ldots,A_k)$ depends only on the state configurations of the DFAs $A_1,\ldots, A_k$ and (only) on the transition function of $a$ in each of the DFAs $A_1,...,A_k$ .
  \end{itemize}
\end{definition}
\begin{example}\label{ex-sqrt}
  For any DFA $A=(\Sigma,Q,i,F,\delta)$, we define \[\mathfrak{Sqrt}(A)=(\Sigma,Q^{Q},\mathrm{Id}_Q,\{\phi\in Q^Q| \phi(\phi(i))\in F\},\delta'),\] where for any $a\in\Sigma$,
  $\delta'^a(\phi)=\delta^a\circ\phi$.
  The modifier $\mathfrak{Sqrt}$ describes the classical construction on DFA for the square root operation on languages \cite{Mas70}, \ie for all DFA $A$, $ \sqrt{\mathrm L(A)}=\mathrm L(\mathfrak{Sqrt}(A))$. In Figure \ref{sqrt}, $[ij]$ represents the function $\phi$ such that $\phi(0)=i$ and $\phi(1)=j$.
  \begin{figure}[H]
    \begin{minipage}[b]{0.49\textwidth}
      \centering
        \captionsetup{justification=centering}
        \begin{tikzpicture}[node distance=2cm]
          \node[state, initial above] (p1) {$0$};
          \node[state, accepting] (p2) at (2,0){$1$};
          \path[->]
          (p1) edge[bend left] node {$a,b$} (p2)
          (p2) edge[bend left] node {$a$} (p1)
          (p2) edge[loop right] node {$b$} (p2);
        \end{tikzpicture}
        \caption{A DFA $A$.}\label{A DFA}
    \end{minipage}
    \begin{minipage}[b]{0.49\textwidth}
      \centering
      \captionsetup{justification=centering}
        \begin{tikzpicture}[node distance=2cm]
          \node[state, initial] (p1) {$[01]$};
          \node[state] (p2) at (2,0) {$[10]$};
          \node[state] (p3) at (1,3) {$[00]$};
          \node[state,accepting] (p4) at (1,2) {$[11]$};
          \path[->]
          (p1) edge[bend left] node {$a$} (p2)
          (p2) edge[bend left] node {$a$} (p1)
          (p3) edge[bend left] node {$a,b$} (p4)
          (p4) edge[bend left] node {$a$} (p3)
          (p1) edge node {$b$} (p4)
          (p2) edge[swap] node {$b$} (p4)
          (p4) edge[loop right] node {$b$} (p4);
        \end{tikzpicture}
      \caption{The DFA $\mathfrak{Sqrt}(A)$.}\label{sqrt}
    \end{minipage}
  \end{figure}
\end{example}

\begin{example}\label{ex-xor}
  For any DFA $A=(\Sigma,Q_1,i_1,F_1,\delta_1)$ and $B=(\Sigma,Q_2,i_2,F_2,\delta_2)$, we define
  \[\mathfrak{Xor}(A,B)=(\Sigma,Q_1\times Q_2,(i_1,i_2),(F_1\times(Q_2\setminus F_2)\cup (Q_1\setminus F_1)\times F_2),(\delta_1,\delta_2))\]
  The modifier $\mathfrak{Xor}$ describes a construction  associated to the  symmetrical difference, \emph{i.e} for all DFAs $A$ and $B$, $\mathrm L(A)\oplus \mathrm L(B)=\mathrm L(\mathfrak{Xor}(A,B))$.
\end{example}

\begin{definition}
  A $k$-modifier $\mathfrak m$ is \emph{$1$-uniform} if, for every pair of $k$-tuples of DFAs $(A_1,\ldots,A_k)$ and $(B_1,\ldots,B_k)$ such that $\mathrm{L}(A_j)=\mathrm{L}(B_j)$ for all $j\in \{1,\dots,k\}$, we have $\mathrm{L}(\mathfrak m(A_1,\ldots,A_k))=\mathrm{L}(\mathfrak m(B_1,\ldots,B_k))$. In that case, there exists a regular operation $\otimes_{\mathfrak m}$ such that, for all $k$-tuples of DFAs $(A_1,\ldots,A_k)$, $\otimes_{\mathfrak m}(\mathrm L(A_1),\ldots,\mathrm L(A_k))=\mathrm{L}(\mathfrak m(A_1,\ldots,A_k))$. We say that $\mathfrak m$ \emph{describes} the operation $\otimes_{\mathfrak m	}$.
\end{definition}
We easily check that, for modifiers,  $1$-uniformity is stable by composition.
\begin{claim}
  Let $\mathfrak m_1$ and $\mathfrak m_2$ be respectively a $j$-modifier and a $k$-modifier describing, respectively, operations $\otimes_1$ and $\otimes_2$.  The modifier $\mathfrak m_1\circ_p\mathfrak m_2$
  describes $\otimes_1\circ_p\otimes_2$.
\end{claim}

The correspondence between $1$-uniform modifiers and $1$-uniform operations is stated in the following Theorem and proved in \cite{CHL19}.
\begin{theorem}
  A $k$-ary operation $\otimes$ is $1$-uniform if and only if there exists a $k$-modifier $\mathfrak m$ such that $\otimes=\otimes_{\mathfrak m}$.
\end{theorem}
\subsection{Functional notations}


When there is no ambiguity, for any symbol $\mathtt{X}$ and any integer $k$ given by the context, we  write $\underline{\mathtt{X}}$ for  $(\mathtt{X}_1,\cdots, \mathtt{X}_k)$. The number $k$ will often be the arity of the regular operation or of the modifier we are considering.

From Definition \ref{def-mod}, any $k$-modifier $\mathfrak m$ can be seen as a  $4$-tuple of mappings $(\mathfrak Q,\mathfrak i,\mathfrak f,\mathfrak d)$ acting on $k$ DFAs $\underline{A}$ with $A_j=(\Sigma, Q_j,i_j,F_j,\delta_ j)$ to build a DFA $\mathfrak m\underline{A}=(\Sigma,Q,i,F,\delta)$, where 
\[Q=\mathfrak Q(\underline{Q},\underline{i},\underline{F})\; i=\mathfrak i(\underline{Q},\underline{i},\underline{F}),\; F=\mathfrak f(\underline{Q},\underline{i},\underline{F}) \text{ and }\forall a\in \Sigma,\ \delta^a=\mathfrak d(\underline{i},\underline{F},\underline{\delta^a}).\]
For the sake of clarity,  we do not write  explicitly the domains of the $4$-tuple of mappings but the reader can derive them easily from the above equalities. Notice that we do not need to point out explicitly the dependency of $\mathfrak d$ on $\underline{Q}$ because the information is already contained in $\underline{\delta^a}$. Notice also that $\delta^a$ indeed only depends on the $k$-tuple of transition functions $\underline\delta^a$ and not on any other transition functions. We identify modifiers and such $4$-tuples of mappings.
\\
Below we revisit the definition of $\mathfrak{Sqrt}$ and $\mathfrak{Xor}$ according to this formalism.
\begin{example}\label{ex-sqrt2}
  $\mathfrak{Sqrt}=(\mathfrak{Q},\mathfrak i,\mathfrak f,\mathfrak d)$ where
    \[\mathfrak Q(Q,i,F)=Q^Q,\; \mathfrak i(Q,i,F)=\mathrm{Id}_Q, \; \mathfrak f(Q,i,F)=\{\phi \mid \phi(\phi(i))\in F\},\; \mathfrak d(i,F,\phi)(\psi)=\phi\circ\psi\]
\end{example}
\begin{example}\label{ex-xor2}
  $\mathfrak{Xor}=(\mathfrak{Q},\mathfrak i,\mathfrak f,\mathfrak d)$ where
\[
  \begin{array}{l}
    \mathfrak Q((Q_1,Q_2),(i_1,i_2),(F_1,F_2))=Q_1\times Q_2,\;\mathfrak i((Q_1,Q_2),(i_1,i_2),(F_1,F_2))=(i_1,i_2),\\
    \mathfrak f((Q_1,Q_2),(i_1,i_2),(F_1,F_2))=F_{1}\times(Q_{2}\setminus F_{2})\cup (Q_{1}\setminus F_{1})\times F_{2},\\
    \mathfrak d((i_1,i_2),(F_1,F_2),(\delta_1,\delta_2))=(\delta_1,\delta_2).
  \end{array}
\]
\end{example}
\section{Friendly modifiers and friendly operations}\label{sect-friend}
\subsection{Friendly modifiers}
We first define friendly modifiers, and eventually give a characterisation of the operations described by friendly modifiers.
\begin{definition}\label{def-friendly}
  We say that a $k$-modifier $\mathfrak{m}=(\mathfrak Q,\mathfrak i,\mathfrak f,\mathfrak d)$ is \emph{friendly} if, for any $k$-tuple of finite sets $\underline Q$, any $\underline F$ such that $F_j\subseteq Q_j$ for all $j$, any $\underline i\in Q_1\times\cdots\times Q_k$, and any $\underline \phi,\underline \psi \in Q_1^{Q_1}\times\cdots\times Q_k^{Q_k}$,\[\mathfrak d(\underline i,\underline F,(\phi_1\circ\psi_1,\ldots,\phi_k\circ\psi_k))=\mathfrak d(\underline i,\underline F,\underline \phi)\circ\mathfrak d(\underline i,\underline F,\underline \psi)\]
\end{definition}
The idea of the definition is that $\mathfrak d$ should be a morphism for its third coordinate. For instance, the modifiers $\mathfrak{Sqrt}$ and $\mathfrak{Xor}$ are friendly. It is easy to check the following property of stability.
\begin{proposition}\label{prop-frcomp}
  Friendly modifiers are stable by composition.
\end{proposition}
\subsection{Standard friendly modifiers}

We define the class \emph{standard} friendly modifiers. The main idea is that to any friendly $1$-uniform modifier is associated  a standard friendly modifier that is another $1$-uniform modifier describing the same operation. Additionally, Theorem \ref{th-frmod}, proven later in Section \ref{4.4}, shows that any  operation described by a friendly modifier is also  described by a unique standard friendly modifier. In other words, standard modifiers is canonical form for every $1$-uniform friendly modifier describing the same operation.
\begin{definition}\label{def-stan}
  We say that a $k$-modifier $\mathfrak m=(\mathfrak Q,\mathfrak i,\mathfrak f ,\mathfrak d)$ is \emph{standard} if
  \begin{itemize}
  \item $\mathfrak Q(\underline Q,\underline i,\underline F)=Q_1^{Q_1}\times\cdots\times Q_k^{Q_k}$
  \item $\mathfrak i(\underline Q,\underline i,\underline F)=(\mathrm{Id}_{Q_1},\ldots,\mathrm{Id}_{Q_k})$
  \item $\mathfrak d(\underline i,\underline F,\underline \phi)(\underline \psi)= (\phi_1\circ\psi_1,\ldots,\phi_k\circ \psi_k)$
  \end{itemize}
\end{definition}
We easily check that a standard $k$-modifier is friendly. Notice that a friendly standard $k$-modifier $\mathfrak m=(\mathfrak Q,\mathfrak i,\mathfrak f,\mathfrak d)$ is entirely defined by its third coordinate $\mathfrak f$.
\begin{definition}
  Let $\mathfrak m=(\mathfrak Q,\mathfrak i,\mathfrak f,\mathfrak d)$ be a friendly $k$-modifier. We denote by $\cf{\mathfrak m}$ the friendly standard $k$-modifier such that $\cf{\mathfrak f}(\underline Q,\underline i,\underline F)=\{\underline \phi \mid \mathfrak d(\underline i,\underline F,\underline \phi)(\mathfrak i(\underline Q,\underline i,\underline F)) \in \mathfrak f(\underline Q,\underline i,\underline F)\}$.
\end{definition}
\begin{example}
  Let $\mathfrak{C}$ be the $1$-modifier $(\mathfrak Q_{\mathfrak{C}},\mathfrak i_{\mathfrak{C}},\mathfrak f_{\mathfrak{C}},\mathfrak d_{\mathfrak{C}})$, with
  \[\mathfrak Q_{\mathfrak{C}}(Q)=Q,\; \mathfrak i_{\mathfrak{C}}(Q,i,F)=i,\; \mathfrak f_{\mathfrak{C}}(Q,i,F)=Q\setminus F,\; \mathfrak d_{\mathfrak{C}}(i,F,d)=d.\]
  The modifier $\mathfrak C$ is friendly. Furthermore, as it follows the classical construction of the complement of DFAs, we easily see that it is $1$-uniform and describes the regular operation complement. Figures \ref{first automata},\ref{second automata} and \ref{third automata} describe the effect of $\mathfrak{C}$ and of $\cf{\mathfrak C}$ on a DFA $A$. In Figure \ref{third automata}, $[ij]$ represents the function $\phi$ such that $\phi(0)=i$ and $\phi(1)=j$.
  \begin{figure}[H]
    \begin{minipage}[b]{0.32\linewidth}
      \centering
      \captionsetup{justification=centering}
        \begin{tikzpicture}[node distance=2cm]
          \node[state, initial above,accepting] (p1) {$0$};
          \node[state] (p2) at (2,0){$1$};
          \path[->]
          (p1) edge[bend left] node {$a,b$} (p2)
          (p2) edge[bend left] node {$a$} (p1)
          (p2) edge[loop right] node {$b$} (p2);
        \end{tikzpicture}
      \caption{The DFA $A$.}\label{first automata}
    \end{minipage}
    \begin{minipage}[b]{0.32\linewidth}
      \centering
      \captionsetup{justification=centering}
        \begin{tikzpicture}[node distance=2cm]
          \node[state, initial above] (p1) {$0$};
          \node[state, accepting] (p2) at (2,0){$1$};
          \path[->]
          (p1) edge[bend left] node {$a,b$} (p2)
          (p2) edge[bend left] node {$a$} (p1)
          (p2) edge[loop right] node {$b$} (p2);
        \end{tikzpicture}
      \caption{The DFA $\mathfrak {C}A$.}\label{second automata}
    \end{minipage}
    \begin{minipage}[b]{0.32\linewidth}
      \centering
      \captionsetup{justification=centering}
        \begin{tikzpicture}[node distance=2cm]
          \node[state, initial] (p1) {$[01]$};
          \node[state,accepting] (p2) at (2,0) {$[10]$};
          \node[state] (p3) at (1,3) {$[00]$};
          \node[state,accepting] (p4) at (1,2) {$[11]$};
          \path[->]
          (p1) edge[bend left] node {$a$} (p2)
          (p2) edge[bend left] node {$a$} (p1)
          (p3) edge[bend left] node {$a,b$} (p4)
          (p4) edge[bend left] node {$a$} (p3)
          (p1) edge node {$b$} (p4)
          (p2) edge[swap] node {$b$} (p4)
          (p4) edge[loop right] node {$b$} (p4);
        \end{tikzpicture}
      \caption{The DFA $\cf{\mathfrak C}A$.}\label{third automata}
    \end{minipage}
  \end{figure}
\end{example}

\begin{lemma}\label{lemma-can}
  For any $1$-uniform friendly $k$-modifier $\mathfrak m$,  the standard modifier $\cf{\mathfrak m}$ describe the same operation as $\mathfrak m$.
\end{lemma}
\begin{proof}
  Let $\mathfrak m=(\mathfrak Q,\mathfrak i,\mathfrak f,\mathfrak d)$ be a $1$-uniform friendly $k$-modifier. We show that $\mathfrak m$ and $\cf{\mathfrak m}$ describe the same operation, which proves that $\cf{\mathfrak m}$ is $1$-uniform. Let $\underline A$ be a $k$-tuple of DFA such that $A_j=(\Sigma,Q_j,i_j,F_j,\delta_j)$.
  
  A word $a_1a_2\ldots a_l$ is in $\mathrm L(\mathfrak m \underline A)$ if and only if \[ \mathfrak d(\underline i,\underline F,\underline\delta^{a_1a_2\ldots a_l})(\mathfrak i(\underline Q,\underline i,\underline F))=(\mathfrak d(\underline i,\underline F,\underline \delta^{a_l})\circ \mathfrak d(\underline i,\underline F,\underline \delta^{a_{l-1}}) \circ \ldots \circ \mathfrak d(\underline i,\underline F,\underline \delta^{a_{1}}))(\mathfrak i(\underline Q,\underline i,\underline F)) \in \mathfrak f(\underline Q,\underline i,\underline F).\] Equivalently, \[(\cf{\mathfrak d}(\underline i,\underline F,\underline \delta^{a_l})\circ \cf{\mathfrak d}(\underline i,\underline F,\underline \delta^{a_{l-1}}) \circ \ldots \circ \cf{\mathfrak d}(\underline i,\underline F,\underline \delta^{a_{1}}))((\mathrm{Id}_{Q_1},\ldots,\mathrm{Id}_{Q_k}))=\underline \delta^{a_1a_2\ldots a_l} \in \cf{\mathfrak f}(\underline Q,\underline i,\underline F).\] But this last statement is equivalent to $a_1a_2\ldots a_l\in\mathrm L(\cf{\mathfrak m}\underline A)$. So $\mathrm L(\cf{\mathfrak m}\underline A)=\mathrm L(\mathfrak m\underline A)$.\cqfd
\end{proof}

We denote by $\mathcal M_k$ the set of $1$-uniform friendly standard $k$-modifiers.

\subsection{Characteristic functions}
As a standard friendly modifier $(\mathfrak Q,\mathfrak i,\mathfrak f,\mathfrak d)$ is entirely characterized  by  the map $\mathfrak f$, which governs the final states of the output DFA, we first show a regularity property on these final states when the modifier is $1$-uniform. To that aim, we associate to every state of an output DFA a characteristic function, in such a way that any two states associated to the same characteristic function have the same finality. These characteristic functions are represented by $k$-tuples of eventually periodic sequences with values in $\{0,1\}$.

Let $\mathcal{U}_k$ be the set of all $k$-tuples $\underline u$ where each $u_j$ is an eventually periodic sequence with values in $\{0,1\}$. To simplify notation, for all $(j,p)\in\{1,\ldots,k\}\times\mathbb N$, we identify $(u_j)_p$ with $u_{j,p}$.
\begin{definition}\label{def-char}
  Let $\underline \phi\in Q_1^{Q_1}\times\cdots\times Q_k^{Q_k}$. We denote by $\chi_{\underline i,\underline F}^{\underline \phi}$ the $k$-tuple of sequences $\underline u\in\mathcal U_k$ where, for any $p\in \mathbb N$ and any $j\in \{1,\ldots,k\}$,
  $u_{j,p}= 1$ if $\phi_j^{p}(i_j)\in F_j$ and  $u_{j,p}= 0$ otherwise ;
    with the notation $\phi_j^p=\underbrace{\phi_j\circ\ldots\circ\phi_j}_{\text p \;times}$. We say that $\chi_{\underline i,\underline F}^{\underline \phi}$ is the $\emph{charateristic sequence}$ of $\underline\phi$ in the state configuration $(\underline Q,\underline i,\underline F)$.
\end{definition}
Notice that, in the above definition, we have $u\in \mathcal U_k$ because  $\phi_j^{p}(i_j)$ is eventually periodic, since $\phi_j$ is a function from a finite set into a finite set.
\begin{example}
  As represented in Figures \ref{first sc}, \ref{second sc}, let $(Q_1,Q_2)=(\{0,1\},\{0,1\})$, $(i_1,i_2)=(0,0)$, $(F_1,F_2)=(\{1\},\{0\})$, $\phi_1(0)=1, \phi_1(1)=0$, $\phi_1=\phi_2$ and $u=\chi_{(i_1,i_2),(F_1,F_2)}^{\phi_1,\phi_2}$. We have, for all $(j,p)\in\{1,2\}\times \mathbb N$, $u_{j,p}=1$ if and only if $p+j$ is even.
    \begin{figure}[H]
    \begin{minipage}[b]{0.49\linewidth}
      \begin{center}
        \begin{tikzpicture}[node distance=2cm]
          \node[state, initial above] (p1) {$0$};
          \node[state,accepting] (p2) at (2,0){$1$};
          \path[->]
          (p1) edge[bend left] node {} (p2)
          (p2) edge[bend left] node {} (p1);
        \end{tikzpicture}
      \end{center}
      \caption{A representation of $(Q_1,i_1,F_1)$ with function $\phi_1$}\label{first sc}
    \end{minipage}
    \begin{minipage}[b]{0.49\linewidth}
      \begin{center}
        \begin{tikzpicture}[node distance=2cm]
          \node[state, initial above, accepting] (p1) {$0$};
          \node[state] (p2) at (2,0){$1$};
          \path[->]
          (p1) edge[bend left] node {} (p2)
          (p2) edge[bend left] node {} (p1);
        \end{tikzpicture}
      \end{center}
      \caption{A representation of $(Q_2,i_2,F_2)$ with function $\phi_2$}\label{second sc}
    \end{minipage}
  \end{figure}
\end{example}
Recall that, if $\underline A$ is any $k$-tuple of DFA such that the set of states of $A_j$ is $Q_j$, the set of the states of $\mathfrak m\underline A$ is $Q_1^{Q_1}\times\cdots\times Q_k^{Q_k}$ when $\mathfrak m$ is standard. The next proposition expresses the fact that states with the same characteristic sequence have the same finality.
\begin{proposition}\label{lemma-root}
  Let $\mathfrak m=(\mathfrak Q,\mathfrak i,\mathfrak f,\mathfrak d)$ be a $1$-uniform standard friendly $k$-modifier. Let $(\underline Q,\underline i,\underline F)$ and $(\underline Q',\underline i',\underline F')$ be any two state configurations, $\underline \phi\in Q_1^{Q_1}\times\cdots\times Q_k^{Q_k}$ and $\underline \phi' \in {Q'_1}^{Q'_1}\times\cdots\times {Q'_k}^{Q'_k}$. If $\chi_{\underline i,\underline F}^{\underline \phi}=\chi_{\underline i,\underline F}^{\underline \phi'}$, then $\underline \phi\in \mathfrak f(\underline Q,\underline i,\underline F)$ if and only if $\underline \phi'\in \mathfrak f(\underline Q',\underline i',\underline F')$.
\end{proposition}
\begin{proof}
  Suppose that $\chi_{\underline i,\underline F}^{\underline \phi}=\chi_{\underline i,\underline F}^{\underline \phi'}$. Because of the symmetry between $\phi$ and $\phi'$ in the lemma, we only need to show that, if $\underline \phi\in \mathfrak f(\underline Q,\underline i,\underline F)$ then $\underline \phi'\in \mathfrak f(\underline Q',\underline i',\underline F')$. Let us therefore suppose that $\underline \phi\in \mathfrak f(\underline Q,\underline i,\underline F)$.
  Let $\underline A$ and $\underline A'$ be $k$-tuples of DFAs with $A_\ell=(\{a\},Q_\ell,i_\ell,F_\ell,\alpha_l)$ and $A'_\ell=(\{a\},Q'_\ell,i'_\ell,F'_\ell,\alpha'_\ell)$ such that $\alpha^a=\phi$ and $\alpha'^a=\phi'$. Since $\chi_{\underline i,\underline F}^{\underline \phi}=\chi_{\underline i',\underline F'}^{\underline \phi'}$, for any $\ell\in\{1,\ldots,k\}$, $\phi_\ell^p(i_\ell)\in F_\ell$ if and only if ${\phi'}_\ell^p(i'_\ell)\in F'_l$. Furthermore, ${\alpha_\ell}^{a^p}=\phi_\ell^p$ and ${\alpha'_\ell}^{a^p}={\phi'_\ell}^p$. Therefore, for any $\ell\in\{1,\ldots,k\}$, $\mathrm{L}(A_\ell)=\mathrm{L}(A'_\ell)$. Since $\mathfrak d (\underline i,\underline F,\underline \alpha^a)(\mathrm{Id}_{Q_1},\ldots,\mathrm{Id}_{Q_k})=\underline \phi$ and $\underline \phi\in \mathfrak f(\underline Q,\underline i,\underline F)$, we have $a\in\mathrm L(\mathfrak m\underline A)$. Furthermore, $\mathfrak m$ is $1$-uniform, and so we have $a\in \mathrm L(\mathfrak m\underline A')$. This implies that $\underline \phi'=\mathfrak d (\underline i',\underline F',\underline \alpha'^a)(\mathrm{Id}_{Q'_1},\ldots,\mathrm{Id}_{Q'_k})\in \mathfrak f(\underline Q',\underline i',\underline F')$, which concludes our proof.\cqfd
\end{proof}
The above result invites us to represent the third coordinate $\mathfrak f$ of a standard friendly modifier by a set of characteristic functions. In fact, Proposition \ref{th-diag}, proven in Section \ref{4.4}, shows that there is a one-to-one correspondence between standard friendly modifiers and subsets of $\mathcal U_k$. Therefore, we now define an application $\mathbf{mod}$ that allows us to compute a standard friendly $k$-modifier from any subset of $\mathcal U_k$.
\begin{definition}\label{def-root}
  For any $E\subseteq \mathcal U_k$, we denote by $\mathbf{mod}(E)$ the friendly standard modifier $(\mathfrak Q,\mathfrak i,\mathfrak f,\mathfrak d)$ with, for all state configurations $(\underline Q,\underline i,\underline F)$, $\mathfrak f(\underline Q,\underline i,\underline F)=\left\{\underline \phi\in Q_1^{Q_1}\times\cdots\times Q_k^{Q_k}\;\mid\; \chi_{\underline i,\underline F}^{\underline \phi}\in E\right\}$.
\end{definition}
As a corollary of Proposition \ref{lemma-root}, any $1$-uniform friendly standard $k$-modifer can be constructed this way from some subset of $\mathcal U_k$. In other words,
\begin{corollary}\label{cor-root}
  The set of $1$-uniform friendly standard $k$-modifiers $\mathcal M_k$ is a subset of the image of $\mathbf{mod}$.
\end{corollary}
\begin{proof}
  Let $\mathfrak m=(\mathfrak Q,\mathfrak i,\mathfrak f,\mathfrak d)$ be a $1$-uniform friendly standard $k$-modifier. Let $E$ be the set of all sequences $u\in\mathcal U_k$ such that there exists a state configuration $(\underline Q,\underline i,\underline F)$ and $\underline \phi\in\mathfrak f(\underline Q,\underline i,\underline F)$ with $\chi_{\underline i,\underline F}^{\underline \phi}=u$. For any state configuration $(\underline Q,\underline i,\underline F)$, if $\underline \phi\in Q_1^{Q_1}\times Q_k^{Q_k}$ and $\chi_{\underline i,\underline F}^{\underline \phi}\in E$, then, by Proposition \ref{lemma-root}, $\phi\in f(\underline Q,\underline i,\underline F)$. The converse is obvious, and we have $\mathfrak f(\underline Q,\underline i,\underline F)=\{\underline \phi\in Q_1^{Q_1}\times Q_k^{Q_k}\;\mid\; \chi_{\underline i,\underline F}^{\underline \phi}\in E\}$. Therefore, $\mathfrak m=\mathbf{mod}(E)$.\cqfd
\end{proof}

\subsection{Friendly operations}\label{4.4}
Examples \ref{ex-sqrt}, \ref{ex-xor}, \ref{ex-sqrt2} and \ref{ex-xor2} show that square root and symmetrical difference can be described by friendly modifiers, and therefore by standard friendly modifiers. These constructions naturally extend to any $k$-th root operation and to any $k$-ary boolean operation. Therefore, by Proposition \ref{prop-frcomp}, any composition of a $k$-ary boolean operation and some roots of languages is described by a standard friendly modifier. These operations are not the only ones to fall in the scope of our study. For instance, the operation $\mathbf{Root}$ \cite{KLS05}, defined by $\mathbf{Root}(L)=\bigcup\limits_{p=1}^{+\infty}\sqrt[p]{L}$, may also be described by a friendly modifier. To capture this kind of operations, we extend the notion of boolean operations to infinite arity.

\begin{definition}
 A \emph{boolean function} is a function from $\{0,1\}^{\mathbb N}$ into $\{0,1\}$. Every boolean function $\mathbf{b}$ defines a boolean operation $\boxtimes_{\mathbf{b}}$ producing a language when acting over sequences of languages in the following way : for any sequence of languages $(L_p)_{p\in\mathbb N}$, a word $w$ is in $\boxtimes_{\mathbf b}((L_p)_{p\in\mathbb N})$ if and only if there exists a sequence $v$ in $\{0,1\}^{\mathbb N}$ with $\mathbf{b}(v)=1$ such that, for all $p\in\mathbb N$, $w\in L_p$ if and only if $v_p=1$.
\end{definition}
\begin{example}
Consider the boolean function $\mathbf{b}$ defined by, for any sequence $v$ in $\{0,1\}^{\mathbb N}$ , $\mathbf{b}(v)=1$ if and only if, either for all $p\in \mathbb N$, $v_{p}=1$, or for all $p\in \mathbb N$, $v_{p}=0$. We have, for any sequence of regular languages $(L_p)_{p\in\mathbb N}$, $w\in \boxtimes((L_p)_{p\in\mathbb N})$ if and only if either for all $p\in \mathbb N$, $w\in L_p$, or for all $p\in \mathbb N$, $w\notin L_p$. This assertion translates into an equation as $\boxtimes((L_p)_{p\in\mathbb N})=\bigcap\limits_{p=0}^{+\infty}L_p\cup \bigcap\limits_{p=0}^{+\infty}L_p^{c}$.
\end{example}
We now have the tools to define friendly operations as the composition of a boolean operation and some roots of languages, and we show in Proposition \ref{th-diag} that there is a one-to-one correspondance between friendly operations, $1$-uniform standard friendly modifiers and subsets of $\mathcal U_k$.

\begin{definition}
  A $k$-ary operation over regular languages $\otimes$ is friendly if there exists a boolean operation $\boxtimes$ such that, for any $k$-tuples of regular languages $\underline L$, \[\otimes(\underline L)=\boxtimes(\sqrt[0]{L_1},\sqrt[0]{L_2},\ldots,\sqrt[0]{L_k},\sqrt[1]{L_1},\sqrt[1]{L_2},\ldots,\sqrt[1]{L_k},\ldots,\sqrt[p]{L_1},\sqrt[p]{L_2}\ldots,\sqrt[p]{L_k},\ldots).\]
\end{definition}
Recall that $\sqrt[0]{L}=\Sigma^{*}$ if 
\begin{definition}\label{def-scal}
  Let $\underline u\in \mathcal U_k$. For any $k$-tuple of regular languages $\underline L$, we define $\langle \underline u,\underline L \rangle=\bigcap\limits_{(p,j)\in \mathbb N\times\{1,\ldots,k\}}E_{p,j}$ where
  $ E_{p,j}=\sqrt[p]{L_j}$ if and
  $E_{p,j}=\sqrt[p]{L_j}^c$ otherwise.\\
  We denote by $\langle \underline u,\cdot\rangle $ the $k$-ary operation over regular languages such that, for any $k$-tuple of regular languages $\underline L$, $\langle \underline u,\cdot\rangle (\underline L)=\langle \underline u,\underline L\rangle $.
\end{definition}
  
\begin{example}
  Let $\underline u\in\mathcal U_2$ be such that $u_{p,j}=1$ if and only if $p+j$ is even. Then, for any two regular languages $L_1$ and $L_2$, \[\langle \underline u,(L_1,L_2)\rangle =(\sqrt[0]{L_1}^c\cap\sqrt[1]{L_1}\cap \sqrt[2]{L_1}^c \cap \sqrt[3]{L_1} \cap \sqrt[4]{L_1}^c\cap\ldots)\bigcap (\sqrt[0]{L_2}\cap\sqrt[1]{L_2}^c\cap \sqrt[2]{L_2} \cap \sqrt[3]{L_2}^c \cap \sqrt[4]{L_2}\cap\ldots)\]
\end{example}
We denote by $\mathcal O_k$ the set of $k$-ary friendly operations.
\begin{definition}
  Let $\mathbf{op}$ be the application from $2^{\mathcal U_k}$ into $\mathcal O_k$ such that, for any $E\subseteq \mathcal U_k$, $\mathbf{op}(E)$ denotes the friendly $k$-ary operation $\bigcup\limits_{u\in E}\langle u,\cdot\rangle$.
\end{definition}
Notice that $\langle \underline u,\underline L \rangle$ is the set of words $w$ such that, for all $(j,p)\in \{1,\ldots,k\}\times\mathbb N\}$, $w\in \sqrt[p]{L_j}$ if and only if $u_{p,j}=1$.
  By remarking that, if $L$ is a regular language then, for any word $w$, the sequence $u=(u_{i})_{i\in\mathbb N}$ satisfying $u_{i}=1$ if $w\in\sqrt[i]L$ belongs to $\mathcal U_{1}$, we show that some of the expected operations can be simulated images by $\mathbf{op}$ of some subset of $\mathcal U_{k}$
  \begin{proposition}\label{prop-root}
	  Any finite composition of  roots, union, intersection and complement acts on regular languages as an operator $\mathbf{op}(E)$ for some $E\in2^{\mathcal U_{k}}$. 
  \end{proposition}
  \begin{proof}
   Let  $L$ be a regular language. We define, for any word $w$, the sequence $u(w)=(u_{i}(w))_{i\in\mathbb N}$ such that $u_{i}(w)=1$ if $w^{i}\in L$ and $0$ otherwise. 
  Since the set of the quotients $(w^{i})^{-1}L$ is finite, the sequence $((w^{i})^{-1}L)_{i\in\mathbb N}$ is eventually periodic and so the sequence $u(w)$ is also eventually periodic. 
  As a consequence $\mathbf{op}(\{u\in\mathcal U_{1}\mid u_{j}=1\})(L)=\sqrt[j]{L}$ for any regular language $L$.  Indeed, if $w\in\sqrt[j]L$ then $u(w)$ is eventually periodic and $u(w)_{i}=1$. Furthermore, from the construction $w\in\langle u(w),L\rangle$.{}
  In the same way, $\mathbf{op}(\{u\in \mathcal U_{1}\mid u_{1}=0\})(L)=L^{c}$,
   $\mathbf{op}(\{(u,v)\in\mathcal U_{2}\mid u_{1}=1\mbox{ and }v_{1}=1\})(L_{1},L_{2})=L_{1}\cap L_{2}$,
    and  $\mathbf{op}(\{(u,v)\in\mathcal U_{2}\mid u_{1}=1\mbox{ or }v_{1}=1\})(L_{1},L_{2})=L_{1}\cup L_{2}$. 
  By iterating these construction, any $k$-ary operator which is a combination of $\sqrt[i]{\ }$, complement, union and intersection can be simulated on regular languages by the action of an operation $\mathbf{op}(E)$ for some $E\in2^{\mathcal U_{k}}$.\cqfd
\end{proof}
  \begin{example}
	  If $L_1, L_{2}$ and $L_{3}$ are regular languages then we have
	  \[{}
	  (\sqrt[i]L_{1}\cup L_{2})\cap L_{3}^{c}=\mathbf{op}(\{(u,v,w)\in\mathcal U_{3}\mid (u_{i}=1\mbox{ or }v_{1}=1)\mbox{ and }w_{1}=0\})(L_{1},L_{2},L_{3}).
	  \]
  \end{example}
  Notice that when acting on $2^{\Sigma^{*}}$ the operator $\mathbf{op}(\{u\in\mathcal U_{1}\mid u_{j}=1\})$ is distinct from $\sqrt[j]{L}$, but the two operators coincide when acting on regular languages.
  
The following lemma proves that there is a one-to-one correspondence between subsets of $\mathcal U_k$ and $k$-ary friendly operations.
\begin{lemma}\label{lemma-bij}
  The application $\mathbf{op}$ is bijective.
\end{lemma}
\begin{proof}
  We first show that $\mathbf{op}$ is surjective. Let $\mathcal V_k=(\{0,1\}^{\mathbb N})^{k}$ be the set of all $k$-tuples of sequences with values in $\{0,1\}$. Let $\otimes$ be a friendly $k$-ary operation and $\boxtimes_{\mathbf b}$ be a boolean operation such that, for any $k$-tuples of regular languages $\underline L$, \[\otimes=\boxtimes_{\mathbf b}(\sqrt[0]{L_1},\dots,\sqrt[0]{L_k},\sqrt[1]{L_1},\ldots,\sqrt[1]{L_k},\ldots,\sqrt[p]{L_k},\ldots,\sqrt[p]{L_k},\ldots).\]
   Let $E=\{\underline u\in \mathcal U_k \;\mid\; \mathbf{b}(u_{1,0},\dots,u_{k,0},u_{1,1},\dots,u_{k,1},\dots,u_{1,p},\dots,u_{k,p},\dots)=1\}$ and $E'=\{\underline v\in \mathcal V_k \;\mid\; \mathbf{b}(v_{1,0},\dots,v_{k,0},v_{1,1},\dots,v_{k,1},\dots,v_{1,p},\dots,v_{k,p},\dots)=1\}$.\\ We check that $\otimes(\underline L)=(\mathbf{op}(E))(\underline L)$.
For any $k$-tuple of regular languages $\underline L$, we have
  \[\otimes(\underline L)=\bigcup\limits_{\underline v\in E'}\{w\in\Sigma^*\;\mid\; \forall (j,p)\in\{1,\ldots,k\}\times\mathbb N, w\in \sqrt[p]{L_j} \Leftrightarrow v_{j,p}=1\}.\]
  Notice that the union above is over a set  which may involve non eventually periodic sequences. We prove that it is not actually the case.
  If $\underline A$ is a $k$-tuple of DFA with $A_j=(\Sigma,Q_j,i_j,F_j,\delta_j)$ such that, for all $j\in\{1,\ldots,k\}$, $\mathrm L(A_j)=L_j$, then $w\in \sqrt[p]{L_j}$ if and only if $(\delta^w_j)^p(i_j)\in F_j$. Therefore, if there exists a word $w$ and a $k$-tuple of sequences $\underline v\in\mathcal V_k$ such that, for all $(j,p)\in\{1,\ldots,k\}\times\mathbb N$, $w\in \sqrt[p]{L_j}$ if and only if $v_{j,p}=1$, then $(\delta^w_j)^p(i_j)\in F_j$ if and only if $v_{j,p}=1$, which implies that $(v_{j,p})_{p\in\mathbb N}$ is eventually periodic. To summarize, if $\{w\in\Sigma^*\;\mid\; \forall (j,p)\in\{1,\ldots,k\}\times\mathbb N, w\in \sqrt[p]{L_j} \Leftrightarrow v_{j,p}=1\}\neq\emptyset$, then $\underline v\in\mathcal U_k$. We thus have
  \[\otimes(\underline L)=\bigcup\limits_{\underline u\in E}\{w\in\Sigma^*\;\mid\; \forall (j,p)\in\{1,\ldots,k\}\times\mathbb N, w\in \sqrt[p]{L_j} \Leftrightarrow u_{j,p}=1\}=\bigcup\limits_{\underline u\in E}\langle \underline u,\underline L\rangle=(\mathbf{op}(E))(\underline L).\]
  
   We now prove that $\mathbf{op}$ is injective. Let $E,E'\subseteq\mathcal U_k$ and $\underline u\in\mathcal U_k$ such that $\underline u\in E$ and $\underline u\notin E'$.
  Since, for any $j\in\{1,\ldots,k\}$, $(u_{j,l})_{l\in\mathbb N}$ is eventually periodic, the languages $L_j=\{a^p\;\mid\; p\in\mathbb N\land u_{j,p}=1\}$ are regular. 
  We have $a\in\sqrt[p]{L_j}$ if and only if $u_{j,p}=1$. Therefore, from Definition \ref{def-scal}, for any $\underline u'\in\mathcal U_k$, $a\in\langle \underline u',\underline L\rangle$ if and only if $\underline u'=\underline u$. It follows that if $\otimes=\mathbf{op}(E)$ and $\otimes'=\mathbf{op}(E')$, then $a\in \otimes\underline L$ and $a\notin\otimes'\underline L$ because $\underline u\in E\setminus E'$. As a consequence, $\otimes\neq\otimes'$ and $\mathbf{op}$ is injective.
\cqfd
\end{proof}
\begin{example}
	For any regular language $L$, we have \[\mathbf{Root}(L)=\mathbf{op}(\{u\in\mathcal U_{1}\mid\mbox{ there exists } i>0\mbox{ such that }u_{i}=1\})(L)=\bigcup_{i\geq 1}\sqrt[i]L.\]
\end{example}
We now show that any operation described by a friendly modifier is friendly.
\begin{lemma}\label{lemma-friendly}
  Let $E\subseteq\mathcal U_k$, $\mathbf{mod}(E)$ describes $\mathbf{op}(E)$. 
\end{lemma}

\begin{proof}
  Let $\mathfrak m=\mathbf{mod}(E)$ with $\mathfrak m=(\mathfrak Q,\mathfrak i,\mathfrak f,\mathfrak d)$ and let $\otimes$ be the operations described by $\mathfrak m$. Let $\underline A$ be any $k$-tuple of DFA with $A_j=(\Sigma,\underline Q,\underline i,\underline F,\underline \delta)$. A word $a_1\cdots a_n$ is in $\mathrm{L}(\mathfrak m\underline A)$ if and only if \[\underline \delta^{a_1\cdots a_n}=(\mathfrak d(\underline i,\underline F,\underline \delta^{a_n})\circ \mathfrak d(\underline i,\underline F,\underline \delta^{a_{l-1}}) \circ \cdots \circ \mathfrak d(\underline i,\underline F,\underline \delta^{a_{1}}))(\mathrm{Id}_{Q_1},\ldots,\mathrm{Id}_{Q_k}) \in \mathfrak f(\underline Q,\underline i,\underline F).\] Equivalently, by Definition \ref{def-root}, $ \chi_{\underline i,\underline F}^{\underline \delta^{a_1\cdots a_n}}\in E$. But by Definition \ref{def-char}, $\chi_{\underline i,\underline F}^{\underline \delta^{a_1\cdots a_n}}$ is the only function $\underline u$ in $E$ such that, for any $ (p,j)\in \mathbb{N}\times\{1,\ldots,k\}$, ${(\delta_j ^{a_1\cdots a_n})}^p(i_j) \in F_j$ if and only if $u_{(p,j)}=1$. Therefore, by Definition \ref{def-scal}, $a_1\cdots a_n\in \mathrm{L}(\mathfrak m(A_1,\ldots,A_k))$ if and only if there exists $\underline u$ in $E$ such that $a_1\cdots a_n\in \langle \underline u,\underline L\rangle $. We thus have $\otimes(\mathrm{L}(A_1),\ldots,\mathrm{L}(A_k))=\bigcup\limits_{u\in E}\langle \underline u,(\mathrm{L}(A_1),\ldots,\mathrm{L}(A_k))\rangle $ and $\otimes=\mathbf{op}(E)$.\cqfd
\end{proof}

For any $k$-ary $1$-uniform modifier $\mathfrak m$, let $\mathbf{desc}$ be the application from $\mathcal M_k$ to $\mathcal O_k$ such that $\mathbf{desc}(\mathfrak m)$ denotes the regular $1$-uniform operation described by $\mathfrak m$. The main result of this section is that all applications of Figure \ref{fig-diag} are bijections and that the diagram  is commutative.
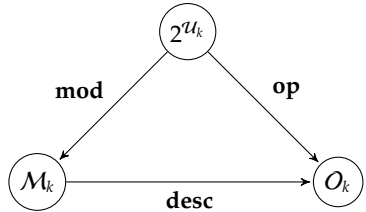
\begin{figure}[H]
  \begin{center}
    \begin{tikzpicture}[node distance=2cm]
      \node[state] (p1) {$\mathcal M_k$};
      \node[state] (p2) at (4,0){$\mathcal O_k$};
      \node[state] (p3) at (2,2){$2^{\mathcal U_k}$};
      \path[->]
      (p3) edge node {$\mathbf{op}$} (p2)
      (p3) edge node[swap] {$\mathbf{mod}$} (p1)
      (p1) edge node[swap] {$\mathbf{desc}$} (p2);
    \end{tikzpicture}
  \end{center}
  \caption{Commutative diagram for $\mathbf{op}$,$\mathbf{mod}$ and $\mathbf{desc}$.}\label{fig-diag}
\end{figure}
In other words,
\begin{proposition}\label{th-diag}
  The application $\mathbf{mod}$ is a bijection from $2^{\mathcal U_k}$ into $\mathcal M_k$, and $\mathbf{op}$ and $\mathbf{desc}$ are bijective. Furthermore, $\mathbf{desc}\circ\mathbf{mod}=\mathbf{op}$.
\end{proposition}
\begin{proof}
  First of all, we already know that $\mathbf{op}$ is a bijection by Lemma \ref{lemma-bij}. Lemma \ref{lemma-friendly} show that a friendly standard $k$-modifier in the image of $2^{\mathcal U_k}$ by $\mathbf{mod}$ is $1$-uniform. Therefore, by Corollary \ref{cor-root}, $\mathbf{mod}$ is a surjection from $2^{\mathcal U_k}$ into the set of $1$-uniform friendly standard $k$-modifiers. By Lemma \ref{lemma-friendly}, the image of $\mathcal M_k$ by $\mathbf{desc}$ is a subset of $\mathcal O_k$. Lemma \ref{lemma-friendly} also proves that $\mathbf{desc}\circ\mathbf{mod}=\mathbf{op}$. As a consequence, $\mathbf{desc}\circ\mathbf{mod}$ is a bijection, and the fact $\mathbf{mod}$ is a surjection implies that both $\mathbf{desc}$ and $\mathbf{mod}$ are bijections.\cqfd
\end{proof}
As an obvious consequence of Proposition \ref{th-diag} and Lemma \ref{lemma-can}, we have :
\begin{theorem}\label{th-frmod}
  Every friendly $k$-ary operation is described by a unique $1$-uniform standard friendly $k$-modifier. Conversely, any $1$-uniform friendly $k$-modifier describes a friendly $k$-ary operation.
\end{theorem}
\section{On the state complexity of friendly operations\label{sec-sc}}
We know that the state complexity of the square root operation \cite{CHLP20} is $sc_{\sqrt{}}(n)=n^n-\binom{n}{2}$, and that it is equal to the state complexity of the operation $\mathbf{Root}$ \cite{KLS05}. 
However, our construction of standard modifiers (Definition \ref{def-stan}) gives us obviously an upper bound of $sc_{\otimes}(n)\leq n^n$ for any unary friendly operation $\otimes$. This raises the question of whether the state complexity of some unary friendly operation reaches this bound and, if not, whether one can give an explicit tight bound. 
 Similar questions arise for the general case of $k$-ary friendly operations with the upper bound $\sc_{\ugh_k}(n_1,\ldots,n_k)\leq\prod\limits_{j=1}^k n_j^{n_j}$  deduced from Definition \ref{def-stan}.
To answer these questions, we use the notion of \emph{monsters} defined in \cite{CHLP20},  the definition of which is recalled  below in a specific case.
\begin{definition}\label{def-mon}
 For any $k$-tuple of positive integers $\underline n$, we denote by $\underline{\mathrm{M}}^{\underline n}$ the $k$-tuple of automata with, for all $j\in\{1,\ldots,k\}$, \[\mathrm M^{\underline n}_j=(\{0,\ldots,n_1-1\}^{\{0,\ldots,n_1-1\}}\times\cdots\times\{0,\ldots,n_k-1\}^{\{0,\ldots,n_k-1\}},\{0,\ldots,n_j-1\},0, \{n_j-1\},\delta_j)\] where, for all $\underline \phi\in \{0,\ldots,n_1-1\}^{\{0,\ldots,n_1-1\}}\times\cdots\times\{0,\ldots,n_k-1\}^{\{0,\ldots,n_k-1\}}$, $\delta_j^{\underline\phi}=\phi_j$.
\end{definition}
\begin{example}
If $\underline n=(2)$, the sequence $\mathrm{M}^{(2)}$ is the DFA below
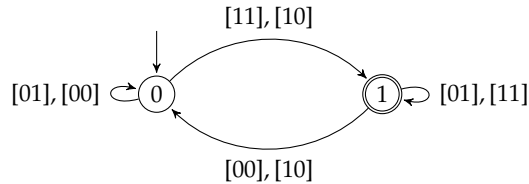
\begin{center}
\begin{tikzpicture}[node distance=3cm]
\node[state,initial above](p0){$0$};{}
\node[state,accepting](p1) [right of=p0] {$1$};
\path[->]
(p0)edge[loop left] node [swap]{$[01],[00]$} (p0)
(p0)edge[bend left] node {$[11],[10]$} (p1)
(p1)edge[loop right] node [swap]{$[01],[11]$} (p1)
(p1)edge[bend left] node{$[00],[10]$}(p0);
\end{tikzpicture}
\end{center}
In this representation, each symbol $[ab]$ denotes the word of the image of a function, \emph{i.e.} the function $g$ from $\{0,1\}$ into $\{0,1\}$ such that $g(0)=a$ and $g(1)=b$. Each function $g$ is associated to a single letter, the transition function of which is $g$.
\end{example}
\subsection{The unary case}\label{una-bound}
We show that the  bound $n^n$ is not tight by the state complexity of friendly operations and we give, and prove, an explicit tight bound.

Consider any subset $E\subseteq \mathcal U_{1}$, and let $\otimes=\mathbf{op}(E)$ and $\mathfrak m=\mathbf{mod}(E)$. Let $A=(\Sigma,Q,i,F,\alpha)$ be a DFA with size $n\in\mathbb N\setminus 0$. We show that the size of the minimal DFA equivalent to $\mathfrak m A$ is at most $n^n-n+1$.
  For all $s,t \in Q$, let $g_{s,t}\in Q^Q$ such that for all $j\in F$, $g_{s,t}(j)=s$ and for all $j\notin F$, $g_{s,t}(j)=t$. The Nerode equivalence on the states of  $\mathfrak mA$ splits the set of the $n^{2}$ functions $g_{s,t}$ into at most $n^{2}-n+1$ classes. The detailed proof of this assertion is given in appendix. This implies that $\sc_{\otimes}(n)\leq n^{n}-n+1$, which gives us the upper bound.

We now show that this bound is tight for $\ugh_1=\mathbf{op}(\{\mathbf 0,\mathbf 0^{1}\})$, where  $\mathbf 0=(0,0,\dots)$ and  $\mathbf 0^{1}=(0,1,1,\dots,1,\dots)$. Notice that $\ugh_{1}(L)=\mathbf{Root}(L)^{c}\cup\{w\in\Sigma^{*}\mid w\in\sqrt[k]{L}\mbox{ for any }k>0\}$ if $\varepsilon\not\in L$ and $\emptyset$ otherwise. Let $\mathfrak w_1=\mathbf{mod}(\{\mathbf 0,\mathbf 0^{1}\})$. 
 We determine a lower bound for the state complexity of $\ugh_1$ by computing the minimal DFA equivalent to $\mathfrak w_1 \mathrm M^n$. By Definition \ref{def-mon} and \ref{def-stan}, we immediately see that the alphabet of $\mathfrak w_1\mathrm M^n$ is $\{0,\ldots,n-1\}^{\{0,\ldots,n-1\}}$, and that every state $\phi$ of $\mathfrak w_1 \mathrm M^n$ is in $\{0,\ldots,n-1\}^{\{0,\ldots,n-1\}}$ and is accessible from its initial state $\mathrm{Id}_{\{0,\ldots,n-1\}}$ by reading the letter $\phi$.
To compute the Nerode equivalence, we need the following  result.
\begin{lemma}\label{lemma-dis}
For any $n\in\mathbb N\setminus 0$, and any $\phi,\psi\in\{0,\ldots,n-1\}^{\{0,\ldots,n-1\}}$  such that $\psi$ is non-constant there exists $\zeta\in \{0,\ldots,n-1\}^{\{0,\ldots,n-1\}}$ such that $\chi^{\zeta\circ\phi}_{0,\{n-1\}}\in\{\mathbf 0,\mathbf 0^{1}\}$ if and only if $\chi^{\zeta\circ\psi}_{0,\{n-1\}}\not\in\{\mathbf 0,\mathbf 0^{1}\}$.
\end{lemma}
\begin{proof}
 The first case we consider is the case where $\phi(0)\neq\psi(0)$. One of the function $\phi$ and $\psi$ is not constant. We assume that it is $\psi$. There exists $i$ such that $\psi(n-1)\neq\psi(i)$. If $\psi(0)\neq\psi(n-1)$ then we set $\zeta(\phi(0))=\zeta(\psi(n-1))=0$ and $\zeta(\psi(0))=n-1$ and this implies $\chi^{\zeta\circ\phi}_{0,\{n-1\}}=\mathbf 0$ and $\chi^{\zeta\circ\psi}_{0,\{n-1\}}=(0,1,0,\dots)\not\in\{\mathbf 0,\mathbf 0^{1}\}$.  Symmetrically, if $\phi(0)\neq\phi(n-1)$ then we obtain our result by permuting the role of $\psi$ and $\phi$ in the previous case. Now suppose that $\phi(0)=\phi(n-1)$ and $\psi(0)=\psi(n-1)$. We set $\zeta(\phi(0))=\zeta(\psi(i))=n-1$, and $\zeta(\psi(0))=i$ which implies  $\chi^{\zeta\circ\phi}_{0,\{n-1\}}=\mathbf 0^{1}$ and $\chi^{\zeta\circ\psi}_{0,\{n-1\}}=(0,0,1,\dots)\not\in\{\mathbf 0,\mathbf 0^{1}\}$.\\
 The first case we consider is the case where $\phi(0)\neq\psi(0)$. One of the function $\phi$ and $\psi$ is not constant. We assume that it is $\psi$. There exists $i$ such that $\psi(n-1)\neq\psi(i)$. If $\psi(0)\neq\psi(n-1)$ then we set $\zeta(\phi(0))=\zeta(\psi(n-1))=0$ and $\zeta(\psi(0))=n-1$ and this implies $\chi^{\zeta\circ\phi}_{0,\{n-1\}}=\mathbf 0$ and $\chi^{\zeta\circ\psi}_{0,\{n-1\}}=(0,1,0,\dots)\not\in\{\mathbf 0,\mathbf 0^{1}\}$.  Symmetrically, if $\phi(0)\neq\phi(n-1)$ then we obtain our result by permuting the role of $\psi$ and $\phi$ in the previous case. Now suppose that $\phi(0)=\phi(n-1)$ and $\psi(0)=\psi(n-1)$. We set $\zeta(\phi(0))=\zeta(\psi(i))=n-1$, and $\zeta(\psi(0))=i$ which implies  $\chi^{\zeta\circ\phi}_{0,\{n-1\}}=\mathbf 0^{1}$ and $\chi^{\zeta\circ\psi}_{0,\{n-1\}}=(0,0,1,\dots)\not\in\{\mathbf 0,\mathbf 0^{1}\}$.\\
  If  $\phi(0)=\psi(0)$ then there exists $j>0$ such that $\phi(j)\neq\psi(j)$. We have $\phi(j)\neq\phi(0)$ or $\psi(j)\neq\psi(0)$. Suppose that $\phi(j)\neq\phi(0)$ (the other case being treated symmetrically). If $j<n-1$ then we set $\zeta(\phi(0))=\zeta(\psi(j))$, and $\zeta(\phi(j))=n-1$. In that case $\chi^{\zeta\circ\phi}_{0,\{n-1\}}=(0,0,1,\dots)\not\in\{\mathbf 0,\mathbf 0^{1}\}$ and $\chi^{\zeta\circ\psi}_{0,\{n-1\}}=\mathbf 0$. Finally if $j=n-1$ then we set $\zeta(\phi(0))=\zeta(\psi(n-1))=n-1$ and $\zeta(\phi(n-1))=0$ which implies $\chi^{\zeta\circ\phi}_{0,\{n-1\}}=(0,1,0,\dots)\not\in\{\mathbf 0,\mathbf 0^{1}\}$ and $\chi^{\zeta\circ\psi}_{0,\{n-1\}}=\mathbf 0^{1}$. This ends the proof.\cqfd
\end{proof}
By Definition \ref{def-root}, the above lemma implies that any two distinct states of $\mathfrak w_1\mathrm M^n$ such that at least one of them is non-constant are distinguishable. Therefore, any non-constant state is distinguishable from every other state and the size of the minimal DFA associated to $\mathfrak w_1\mathrm M^n$ is at least equal to the cardinality of the set of functions of $\{0,\ldots,n-1\}^{\{0,\ldots,n-1\}}$ that are not constant. Thus, for every $n\in\mathbb N\setminus 0$, the size of the minimal DFA equivalent to $\mathfrak w_{1} \mathrm M^{n}$ is at least $n^n-n+1$ , and so we have $\sc_{\ugh_{1}}(n)\geq n^n-n+1$. As a consequence,
\begin{theorem}\label{th-unitary}
  For any friendly operation $\otimes$, $\sc_{\otimes}(n)\leq n^{n}-n+1$ and the bound is tight for $\ugh_{1}$.
\end{theorem}
\begin{proof}
  Consider any subset $E\subseteq \mathcal U$. Let $\otimes=\mathbf{op}(E)$ and $\mathfrak m=\mathbf{mod}(E)$. Let $A=(\Sigma,Q,i,F,\alpha)$ be a DFA with size $n\in\mathbb N\setminus 0$. We show that $\sc(\mathrm{L}(\mathfrak m A))$ is at most $n^n-n+1$. We first suppose that $i\notin F$. 
  For all $s,t \in Q$, let $g_{s,t}\in Q^Q$ such that for all $j\in F$, $g_{s,t}(j)=s$ and for all $j\notin F$, $g_{s,t}(j)=t$. When $t\in F$ we have  $\chi_{i,F}^{g_{s,t}}=\mathbf 0^{1}$ if $s\in F$  and  $\chi_{i,F}^{g_{s,t}}=\mathbf {odd}=(0,1,0,1,\dots,n\mod 2,\dots)$ otherwise. Furthermore if $t\not\in F$ then  $\chi_{i,F}^{g_{s,t}}=\mathbf {0}$.  Let $G=\{g_{s,t}\mid s,t\in Q\}$. We remark $G$ is stable by external composition that is for any $g_{s,t}\in G$, we have $\zeta\circ g_{s,t}=g_{\zeta(s),\zeta(t)}\in G$.
  \\ 
  Let $E_{1}=\{\mathbf 0,\mathbf 0^{1},\mathbf{odd}\}\cap E$ and $E_{2}=\{\mathbf 0,\mathbf 0^{1},\mathbf{odd}\}\setminus E_{1}$.\\
  If $\#E_{1}=0$ (resp. $\#E_{1}=3$) then for any $s,t\in Q$ then, since  $\chi^{ g_{s,t}}_{i,F}\notin E$ (resp. $\chi^{ g_{s,t}}_{i,F}\in E$) , the state $g_{s,t}$ is not final (resp. final). Since, $G$ is stable by external composition, all the states in $G$ are in the same  class for the Nerode equivalence. So $\sc(\mathrm{L}(\mathfrak m A))\leq n^{n}-n^{2}+1\leq n^n-n+1$. \\
%
 If $\#E_{1}=1$ (resp. $\#E_{1}=2$) then we denote by $u$ the unique element of $\#E_{1}$ (resp. $\#E_{2}$). If $u=\mathbf {odd}$ then by remarking that, for any $s,s',t\in Q$ we have  $\chi_{i,F}^{g_{s,t}}=\mathbf{odd}$ if and only if $\chi_{i,F}^{g_{s',t}}=\mathbf{odd}$, the stability of $G$ implies  that any two states $g_{s,t}$ and $g_{s',t}$ are not distinguishable in $\mathfrak mA$. As a consequence,  $\sc(\mathrm{L}(\mathfrak m A))\leq n^{n}-n(n-1)\leq n^n-n+1$. If $u=\mathbf 0^{1}$ then, by remarking that  $\chi_{i,F}^{g_{s,t}}=\mathbf 0^{1}$ if and only if $\chi_{i,F}^{g_{t,s}}=\mathbf 0^{1}$, the stability of $G$ implies that the states $g_{s,t}$ and $g_{t,s}$ are not distinguishable for any $s,t\in Q$. So, we have $\sc(\mathrm{L}(\mathfrak m A))\leq n^{n}-\frac12n(n-1)\leq n^n-n+1$.
  Finally, consider the case where  $u=\mathbf 0$. By remarking that $\chi^{g_{s,s}}_{i,F}=\mathbf 0$, the stability of $G$ implies that  any two  states $g_{s,s}$ and $g_{s',s'}$ are not distinguishable in $\mathfrak mA$. So,  $\sc(\mathrm{L}(\mathfrak m A))\leq n^n-n+1$. \\
  Now assume that $i\in F$.  Then when $s\not\in F$ we have $\chi^{g_{s,t}}_{i,F}=\mathbf 1^{0}=(1,0,0,\dots)$ if $t\not\in F$ and $\chi^{g_{s,t}}_{i,F}=\mathbf{even}=(1,0,1,0,\dots,n+1\mod 2,\dots)$ otherwise. Furthermore if $s\in F$ we have $\chi^{g_{s,t}}_{i,F}=\mathbf 1=(1,1,\dots)$. The proof goes symmetrically to the case where $i\not\in F$, by exchanging the role of the state $s$ and $t$ in the proof, and replacing all the occurrences of $\mathbf 0$ by $\mathbf 1$, all the occurrences of $\mathbf 0^{1}$ by $\mathbf 1^{0}$, and all the occurrences of $\mathbf  {odd}$ by $\mathbf {even}$.\\
  To summarize, in all the cases $\sc(L(\mathfrak m A))\leq n^{n}-n+1$ and so $\sc_{\otimes}(n)\leq n^{n}-n+1$ for any friendly unary operation $\otimes$.\\
  As a consequence of Lemma \ref{lemma-dis}, any non constant state is distinguishable from every other state. Hence, for every $n\in\mathbb N\setminus 0$, the size of the minimal DFA equivalent to $\mathfrak w_{1} \mathrm M^{n}$ is at least $n^n-n+1$ , and so we have $\sc_{\ugh_{1}}(n)\geq n^n-n+1$.
	Indeed, there are $n$ constant functions in $\{0,\ldots,n-1\}^{\{0,\ldots,n-1\}}$, and  every state $\zeta$ of $\mathfrak w_{1}\mathrm M^{n}$ is accessible from the initial state $\mathrm{Id}_{\{0,\ldots,n-1\}}$ by the letter $\zeta$. It follows that $\sc_{\ugh_{1}}(n)=n^{n}-n+1$. This ends the proof.\cqfd
\end{proof}	
\subsection{The general case}
Surprisingly, unlike the unary case, we show that there are friendly operations which state complexity meet the  upper bound $\prod\limits_{j=1}^kn_j^{n_j}$. We exhibit an operation $\ugh_{k}$ and a witness such that the DFA obtained by acting on a witness by its associated standard modifier is minimal. We assume that $k\geq 2$,
  and set $\ugh_{k}=\mathbf{op}(E_{k})$ with $E_{k}=\{\mathbf 0,\mathbf 0^{1}\}^{k}\setminus (\mathbf 0,\dots,\mathbf 0)$.
By 
using Lemma \ref{lemma-dis}, we  prove that, for any $\phi,\psi\in\{0,\ldots,n_1-1\}^{\{0,\ldots,n_1-1\}}\times\cdots\times\{0,\ldots,n_k-1\}^{\{0,\ldots,n_k-1\}}$ with $\phi\neq\psi$, there exists $\underline\zeta\in\{0,\ldots,n_1-1\}^{\{0,\ldots,n_1-1\}}\times\cdots\times\{0,\ldots,n_k-1\}^{\{0,\ldots,n_k-1\}}$ such that $\chi_{(0,\ldots,0),(\{n_{1}-1\},\dots,\{n_{k}-1\})}^{\underline\zeta\circ\underline\phi}\in E_k$ if and only if $\chi_{(0,\ldots,0),(\{n_{1}-1\},\dots,\{n_{k}-1\})}^{\underline\zeta\circ\underline\psi}\notin E_k$. We thus have :
\begin{theorem}\label{th-gen}
  For any $\underline n\in(\mathbb N\setminus 0)^k$, $\sc_{\ugh_k}(\underline n)=\prod\limits_{j=1}^kn_j^{n_j}$ and $(\underline{\mathrm M}^{\underline n})_{\underline n\in(\mathbb N\setminus 0)^k}$ is a witness for $\ugh_k$.
\end{theorem}
\begin{proof}
It suffices to show that for any  $\underline\phi\neq\underline\psi$  there exists $\underline\zeta$ such that $\chi_{(0,\dots,0),(\{n_{1}-1\},\dots,\{n_{k}-1\})}^{\underline\zeta\circ\underline\phi}\in E_{k}$ if and only if $\chi_{(0,\dots,0),(\{n_{1}-1\},\dots,\{n_{k}-1\})}^{\underline\zeta\circ\underline\psi}\not\in E_{k}$.\\
Let $\ell$ such that $\phi_{\ell}\neq\psi_{\ell}$. We have to consider two cases:\\
If both $\phi_{\ell}$ and $\psi_{\ell}$ are constant functions then we set $\zeta_{\ell}(\phi(0))=0$ and $\zeta_{\ell}(\psi(0))=n_{\ell}-1$ and for any $i\neq \ell$ we choose $\zeta_{i}$ as the constant function sending any element to $0$.
 So we have   $\chi_{(0,\dots,0),(\{n_{1}-1\},\dots,\{n_{k}-1\})}^{\underline \zeta\circ\underline \phi}=(\mathbf 0,\dots,\mathbf 0)\not\in E_{k}$ 
 and  $\chi_{(0,\dots,0),(\{n_{1}-1\},\dots,\{n_{k}-1\})}^{\underline \zeta\circ\underline\psi}=(\mathbf 0,\dots,\mathbf 0,\mathbf 0^1,\mathbf 0,\dots,\mathbf 0)\in E_{k}$. \\
 If one of the functions $\phi$ and $\psi$ then we can use Lemma \ref{lemma-dis} and deduce that there exists a function $\zeta$ such that 
  $\chi_{0,\{n_{\ell}-1\}}^{\zeta\circ\phi_{\ell}}\in \{\mathbf 0,\mathbf 0^1\}$ if and only if$\chi_{0,\{n_{\ell}-1\}}^{\zeta\circ\psi_{\ell}}\not\in \{\mathbf 0,\mathbf 0^1\}$. We assume that  $\chi_{0,\{n_{\ell}-1\}}^{\zeta\circ\phi_{\ell}}\in \{\mathbf 0,\mathbf 0^1\}$ (the other case being obtained symmetrically). We choose each $\zeta_{i}$, with $i\neq \ell$, as the constant function sending any element to $n_{i}-1$ and $\zeta_{\ell}=\zeta$. We have
   $\chi_{(0,\dots,0),(\{n_{1}-1\},\dots,\{n_{k}-1\})}^{\underline \zeta\circ\underline\psi}=(\mathbf 0^1,\dots,\mathbf 0^1, \chi_{0,\{n-1\}}^{\zeta\circ\psi},\mathbf 0^1,\dots,\mathbf 0^1) \not\in E_{k}$,
    since $ \chi_{0,\{n_{\ell}-1\}}^{\zeta\circ\psi_{\ell}}\not\in\{\mathbf 0,\mathbf 0^1\}$, 
    and $\chi_{(0,\dots,0),(\{n_{1}-1\},\dots,\{n_{\ell}-1\})}^{\underline \zeta\circ\underline \phi}=(\mathbf 0^1,\dots,\mathbf 0^1, \chi_{0,\{n_{\ell}-1\}}^{\zeta\circ\phi},\mathbf 0^1,\dots,\mathbf 0^1) \in E_{k}$ because  $ \chi_{0,\{n_{\ell}-1\}}^{\zeta\circ\phi_{\ell}}\in\{\mathbf 0,\mathbf 0^1\}$. This ends the proof.\cqfd
    \end{proof}
Notice that in Theorems \ref{th-unitary} and \ref{th-gen}, the size of the alphabet, depending on $\underline n$, is unbounded.  However, the states of the resulting automata are indexed by $k$-tuples of functions and the transitions are completely described by the point to point compositions of transformations. Each monoid of transformation being generated by $3$ elements, we can restrict the alphabet by choosing only letters corresponding to  generators of $\{0,\ldots,n_{1}-1\}^{\{0,\ldots,n_{1}-1\}}\times\cdots\times\{0,\ldots,n_{k}-1\}^{\{0,\ldots,n_{k}-1\}}$, and  simulate the other  transition functions with by composition because the modifier is friendly. So the bounds remain tight for alphabets of size  $3k$.
\section{Conclusion}
We found a tight bound for the state complexity of friendly modifiers, which gives us a tight bound for the state complexity of infinite unions of intersections of roots. It is very probable that reducing the size of the alphabet of friendly modifiers to two would lead to another bound on their state complexity, which remains to be found.\\
Our future works could be to find other interesting classes of modifiers that are stable by composition, so that their study may lead to interesting new results. Furthermore they could begin to paint an interesting general picture that we would use as a basis for finding more general theorems on state complexity.

\bibliography{../COMMONTOOLS/biblio}

\end{document}